\documentclass[a4paper]{amsart}

\usepackage{graphics,amssymb,enumerate}
\usepackage{hyperref}

\newcommand {\bd} {\begin{displaymath}}
\newcommand {\ed} {\end{displaymath}}
\newcommand {\be} {\begin{equation}}
\newcommand {\ee} {\end{equation}}
\newcommand {\bea} {\begin{eqnarray}}
\newcommand {\eea} {\end{eqnarray}}

\def\C{C^{\infty}(M)}

\newtheorem{lemma}{Lemma}

\newcommand{\so}[2]{\ensuremath{\mathfrak{so}(#1,#2)}}

\newcommand{\tr} {\ensuremath{\text{\upshape tr\,}}}
\newcommand{\Set}[2] {\ensuremath{ \{\,#1\,|\,#2\,\}}}
\newcommand{\gln}[2][n]{\ensuremath{\mathfrak{gl}(#1,#2)}}
\newcommand{\nR}{\ensuremath{\mathbb{R}}}

\newcommand{\csH}{\ensuremath{\mathfrak{h}}}
\newcommand{\Ga}{\ensuremath{\alpha}}
\newcommand{\ii}{{\rm i}}
\newlength{\lengthX}
\newcommand{\nin}  {\not\in}
\newcommand{\Gp}{\ensuremath{\pi}}
\newcommand{\Gl}{\ensuremath{\lambda}}
\newcommand{\Gm}{\ensuremath{\mu}}
\newcommand{\Gn}{\ensuremath{\nu}}
\newcommand{\Gx}{\ensuremath{\chi}}
\newcommand{\Def}  [1]{\textit{#1}\index{#1}}

\newcounter{alist}
\newenvironment{alist}
	{\begin{list}
		{ {\rm (\alph{alist})} }
		{
			\usecounter{alist}
		}
	}
	{\end{list}}

\newtheorem{theorem}{Theorem}
\newtheorem{example}{Example}

\Large

\begin{document}
%------------------------------------------------------------------------------%
%\title{  so(p,q)  Toda Systems}
%\author{ Stelios A. Charalambides  and Pantelis A.~Damianou}
%\address{Department of Mathematics and Statistics\\
%University of Cyprus\\
%P.O.~Box 20537, 1678 Nicosia\\Cyprus}
%\email{steliosc@ucy.ac.cy, \ \    damianou@ucy.ac.cy}
%------------------------------------------------------------------------------%
\title{so(p,q) Toda Systems}

\author{Stelios A.~Charalambides}
\email[Stelios A.~Charalambides]{steliosc@ucy.ac.cy}

\author{Pantelis A.~Damianou}
\email[Pantelis A.~Damianou]{damianou@ucy.ac.cy}

\address{Department of Mathematics and Statistics, University of Cyprus, P.O.~Box 20537, 1678 Nicosia, Cyprus}
%------------------------------------------------------------------------------%

\begin{abstract}
 We define an integrable hamiltonian system of Toda type  associated with the real Lie algebra $\so{p}{q}$.  As usual there exists a periodic and a non-periodic version.   We construct, using the root space, two  Lax pair representations  and the associated Poisson tensors. We prove Liouville integrability and examine the multi-hamiltonian structure. The system is a projection of a canonical $A_n$ type Toda lattice via a Flaschka type transformation. It is also obtained via a complex change of variables from the classical Toda lattice.
\end{abstract}

%------------------------------------------------------------------------------%
\maketitle
%------------------------------------------------------------------------------%

%------------------------------------------------------------------------------%
\section{Introduction}
%------------------------------------------------------------------------------%
The Toda lattice, is a Hamiltonian system with an exponential potential:

$$
  H(q_1,\dots, q_N,\,p_1,\dots, p_N)=\sum_{i=1}^N \,\frac12\, p_i^2 +\sum _{i=1}^{N-1} \, e^{ q_i-q_{i+1}}.
$$
It describes a system of $N$ particles on a line, connected by exponential springs.
The Toda lattice is a well-known integrable system with soliton solutions.  Under  Flaschka's change of variable ~\cite{flaschka1}

\be
	a_i  = \frac{1}{2} e^{ \frac{1}{2} (q_i - q_{i+1} ) }  \ \ , \quad \quad
	b_i  = -\frac{1}{2} p_i  \ ,  \label{f1}
\ee

the equations become
\be
\begin{array}{lcl}
 \dot a _i& = & a_i \,  (b_{i+1} -b_i )    \\
   \dot b _i &= & 2 \, ( a_i^2 - a_{i-1}^2 ) \ .
\end{array}
\ee

These  equations can be written in  Lax pair  form $\dot L = [B, L] $, where $L$ is the Jacobi matrix
\begin{equation}
L =
\begin{pmatrix}
	b_1 &  a_1 & 0 & \cdots & \cdots & 0 \cr
	a_1 & b_2 & a_2 & \cdots &    & \vdots \cr
	0 & a_2 & b_3 & \ddots &  &  \cr
	\vdots & & \ddots & \ddots & & \vdots \cr
	\vdots & & & \ddots & \ddots & a_{N-1} \cr
	0 & \cdots & & \cdots & a_{N-1} & b_N
\end{pmatrix}
                        \ ,
\end{equation}
and $B$ is the skew-symmetric part of $L$ in the Lie algebra decomposition lower triangular plus skew-symmetric. Due to the Lax pair, it  follows that the functions	$H_i=\frac{1}{i}\, \tr L^i$ are constants of motion.  Moreover, they are in involution with respect to a Poisson structure, associated to this  Lie algebra decomposition.

We note that
\bd
H_1=\sum_{i=1}^N b_i=-\frac{1}{2} (p_1+p_2+ \dots + p_N)  \ ,
\ed
corresponds to the total momentum and
\bd
H_2=  H(q_1, \dots, q_N, \,  p_1, \dots, p_N) = \frac{1}{2} \sum_{i=1}^N b_i^2+  \sum_{i=1}^{N-1} a_i^2
\ed
is the Hamiltonian.

Consider ${\bf R}^{2N} $  with  coordinates $(q_1, \dots , q_N, p_1, \dots, p_N)$, the
standard symplectic bracket $\{\ , \ \}_0$  and the Flaschka' transformation  $F: {\bf R}^{2N} \to {\bf R}^{2N-1}$ defined by
\bd
 F:  (q_1, \dots, q_N, p_1, \dots, p_N) \to (a_1,  \dots, a_{N-1}, b_1, \dots, b_N) \ .
\ed

There exists a bracket on ${\bf R}^{2N-1} $ which satisfies
\bd
\{f, g \} \circ F  = \{ f \circ F, g \circ F \}_0 \ .
\ed

It  is a bracket  which (up to a constant multiple) is given by
\be
\begin{array}{lcl}
\{a_i, b_i \}& =&-a_i  \\
\{a_i, b_{i+1} \} &=& a_i  \label{poisson1}   \ ;
\end{array}
\ee
all other brackets are zero.
$H_1=b_1+b_2 + \dots +b_N$ is the only Casimir.   The Hamiltonian in this bracket is
 $H_2 = \frac{1}{2}\  { \rm tr}\  L^2$.  We also have involution of invariants,  $ \{  H_i, H_j \}=0$.
We denote this bracket by $\pi_1$.

The quadratic Toda bracket,  $\Pi_2$, due to Adler \cite{adler}, can be used to define a bi-hamiltonian formulation of the system.
It is a Poisson bracket in which the Hamiltonian vector field generated by $H_1$ is the same as the Hamiltonian vector field generated by $H_2$ with respect to the $\pi_1$ bracket. The defining relations are
\be \label{Adler_bracket}
\begin{array}{lcl}
\{a_i, a_{i+1} \}&=&\frac{1}{2} a_i a_{i+1} \\
\{a_i, b_i \} &=& -a_i b_i          \\
\{a_i, b_{i+1} \}&=& a_i b_{i+1}  \\
\{b_i, b_{i+1} \}&=& 2\, a_i^2  \ ;
\end{array}
\ee
all other brackets are zero.
This bracket has ${\rm det} \, L$ as Casimir and
	$H_1 =\tr  L$ is the Hamiltonian.
Furthermore, $\Pi_2$ is compatible with $\pi_1$.
We also have the Lenard relations
\be
\Pi_2 \nabla H_l = \pi_1 \nabla H_{l+1} \ .
\ee

The classical Toda lattice was  generalized to produce  a Toda type system for each
simple Lie algebra. The finite, non--periodic Toda lattice corresponds to a root system of type $A_\ell$. This
generalization is due to Bogoyavlensky \cite{bogo}. These systems were studied extensively in \cite{kostant} where
the solution of the system was connected intimately with the representation theory of simple Lie groups.   We call these systems the Bogoyavlensky-Toda lattices.  They can be described as follows.

\smallskip

Let $\mathfrak{g}$ be any simple Lie algebra over the complex numbers. One chooses a Cartan subalgebra $\mathfrak{h}$ of $\mathfrak{g}$, and a basis $\Pi$ of simple roots for the root system $\Delta$ of $\mathfrak{h}$ in $\mathfrak{g}$. The corresponding set of positive roots is denoted by $\Delta^{+}$.  To each positive root $\alpha$ one can associate a triple $(X_\alpha,X_{-\alpha},H_{\alpha})$ of vectors in $\mathfrak{g}$ which generate a Lie subalgebra isomorphic to $sl_2(\mathbf{C})$. The set $(x_\alpha, x_{-\alpha})_{\alpha \in \Delta^+}\cup (h_\alpha)_{\alpha \in \Pi}$ is basis of $\mathfrak{g}$, called a root basis.  To these data one associates the Lax equation $\dot L=[B,L]$, where $L$ and $B$ are defined as follows:
\begin{eqnarray*}
  L&=&\sum_{i=1}^\ell b_i h_{\alpha_i} + \sum_{i=1}^\ell a_i (x_{\alpha_i}+x_{-\alpha_i}),\\
  B&=&\sum_{i=1}^\ell a_i (x_{\alpha_i}-x_{-\alpha_i}).
\end{eqnarray*}
The affine space  of all elements $L$ of $\mathfrak{g}$ of the above form is the phase space of the Bogoyavlensky-Toda
lattice, associated to $\mathfrak{g}$. The functions which yield the integrability of the system are the $Ad$-invariant
functions on $\mathfrak{g}$, restricted to $M$.

 In this paper we imitate this procedure for the case of a real Lie algebra of type $\so{p}{q}$.  The resulting system is Liouville integrable as well. It turns out that this system can be realized as a projection of a classical type Toda system where some of the terms in the potential are switched from positive to negative. We illustrate  with a small example: The potential of the Toda lattice  is exponential
 \bd V(q_1, \dots , q_{N})=\sum _{i=1}^{N-1} \, e^{ q_i-q_{i+1}} \ . \ed
For example, if $N=5$
 \bd V(q_1, q_2, q_3, q_4, q_5)=e^{ q_1-q_{2}}+e^{ q_2-q_{3}}+e^{ q_3-q_{4}}+e^{ q_4-q_{5}} \ . \ed

 We may modify the potential by changing  some of the plus sings to minuses the system remains integrable. There are $16$ total possibilities for the possible $\pm$ signs in front of the exponentials.  For example one may consider the potential
\bd V(q_1, q_2, q_3, q_4, q_5)=e^{ q_1-q_{2}}+e^{ q_2-q_{3}}-e^{ q_3-q_{4}}+e^{ q_4-q_{5}} \  . \ed

 By projecting this system using  Flaschka's  transformation we  obtain a new system in $(a,b)$ variables and we will give a Lax pair for this system. It turns out that the Lax pair is associated with the real Lie algebra $\so{6}{5}$.
More generally we define an integrable hamiltonian system associated with the real Lie algebra $\so{2m}{2n+1}$.
As usual there exists a periodic and a non-periodic version.

We construct, using the root space, a Lax pair representation and the associated Poisson tensors. We prove Liouville integrability and examine the multi-hamiltonian structure.
The system is bi-hamiltonian with a linear and a quadratic Poisson bracket as in the case of the classical Toda lattice.  The quadratic Poisson bracket is similar to the Adler Toda bracket with some differences on the signs. In fact one can obtain this new bracket via a simple complex change of variables. The system is a projection of a canonical $A_n$ type Toda lattice (with some sign changes on the potential) via a Flaschka type transformation. It turns out that this system is equivalent via a complex change of variables with the classical non-periodic Toda lattice.  As we mentioned, there are a total of $2^N$ such systems but in this paper  we just consider one such potential for simplicity.  Note that the number of systems is related to the construction of Tomei \cite{tomei} of the Toda manifold, i.e. the space of real, symmetric, tridiagonal $n\times n$ matrices with fixed eigenvalues.

We begin by giving a description of the basis of $\so{2m}{2n+1}$, its Cartan subalgebra, its roots and its root space. We end--up with a new set of polynomial equations in the variables $(a,b)$. One can write the equations in Lax pair form ($L(t), B(t)$), which can be described in terms of the root system. We also compute the Casimirs associated with both Poisson brackets and prove the involution of the invariants. Finally we show that this system is equivalent via a simple complex change of variables to the classical Toda system. Using this approach we give a second Lax pair involving complex coefficients and a different set of integrals in involution.

%------------------------------------------------------------------------------%
\section{Bi-hamiltonian systems and Master Symmetries}
%------------------------------------------------------------------------------%
A Poisson bracket  on the space $\C$ of smooth functions on a smooth manifold $M$ is a skew-symmetric, bilinear map,
\begin{equation*}
\{\cdot,\cdot\} : \C \times \C \to \C,
\end{equation*}
that verifies the Jacobi identity and is a biderivation. Thus, $(\C, \{\cdot,\cdot\})$ has the structure of a Lie algebra.
A \emph{Poisson structure} on a smooth manifold $M$ is a Lie algebra structure on $C^{\infty}(M)$ whose the bracket $\{\cdot,\cdot\} : C^{\infty}(M) \times C^{\infty}(M) \to C^{\infty}(M)$ verifies the Leibniz's rule:
\begin{equation*}
\{f,gh\} = \{f,g\}h + g\{f,h\}, \quad \quad \forall \, f,g,h \in C^{\infty}(M).
\end{equation*}
 The Poisson bracket $\{\cdot,\cdot\}$ gives rise to a contravariant antisymmetric tensor field $\pi$  of order $2$ such that $\pi(df,dg) = \{f,g\}$, for $f,g\in C^{\infty}(M)$.  The bivector  $\pi$ is called a \emph{Poisson tensor} and the manifold $(M,\pi)$ a \emph{Poisson manifold}.

A Hamiltonian system is specified by a Poisson bracket $\{\ , \ \}$ together with a $C^{\infty}$ function $H$ called the Hamiltonian.  The vector field has the form
\bd X=\pi dH \ . \ed
A   bi-Hamiltonian system is defined by specifying two Hamiltonian functions $H_1$, $H_2$ satisfying
\begin{equation}
X=\pi_1  \,  d H_2 = \pi_2 \, d  H_1 \ .
\end{equation}

\smallskip
\noindent

If  $\pi_1$  is
symplectic, i.e., the structure matrix $\pi_1$ is invertible,  we call the Poisson pair non-degenerate.  If we assume a non-degenerate pair we make
the following definition: The  recursion operator associated with a non-degenerate pair is the
$(1,1)$-tensor ${\mathcal R}$ defined by
\begin{equation}
{\mathcal R}=\pi_2  \pi_1^{ -1}.
\end{equation}

We have the following result due to Magri \cite{magri}:

\begin{theorem}
 Suppose  that we have a non--degenerate bi-Hamiltonian system on a manifold $M$, whose first
cohomology group is trivial. Then,  there exists a hierarchy of mutually commuting
functions $H_1, H_2, \dots $,  all in involution with respect to both brackets. If we denote by
$\chi_i$ the Hamiltonian vector field generated by $H_i$  with respect to the initial bracket $\pi_1$ then
the $\chi_i$  generate mutually commuting bi-Hamiltonian flows, satisfying the
Lenard recursion relations
\begin{equation}
\chi_{i+j}= \pi_i \,   \nabla H_j  \ ,
\end{equation}
where $\pi_{i}={\mathcal R}^{i-1} \pi_1$ are the higher order Poisson tensors.

\end{theorem}

The definition and basic properties of master symmetries can be found in Fuchssteiner \cite{fuchssteiner}.
Consider a differential equation on a manifold $M$ defined by a vector field $\Gx$.

A vector field $Z$ is a \Def{symmetry} of the equation if
\begin{equation}
	[Z, \chi]=0 \ .
\end{equation}

If $Z$ is time dependent, then a more general condition is
\begin{equation}
	\frac{\partial Z}{\partial t} + [ Z, \Gx]=0 \ .
\end{equation}

A vector field $Z$ is called a \Def{master symmetry} if
\begin{equation}
	[[Z, \Gx], \Gx]=0 \ ,
\end{equation}
but
\begin{equation}
	[Z, \Gx] \not= 0 \ .
\end{equation}

Master symmetries were first introduced by Fokas and Fuchssteiner in \cite{fokas1} in connection with the Benjamin-Ono Equation.

Suppose that we have a bi-Hamiltonian system defined by
	the Poisson tensors $\Gp_1$, $\Gp_2$ and
	the Hamiltonians $H_1$, $H_2$.
Assume that $\Gp_1$ is symplectic.
We define the \Def{recursion operator}
	${\mathcal R} = \Gp_2 \Gp_1^{-1}$,
the higher flows
\begin{equation}
	\Gx_{i} = {\mathcal R}^{i-1} \Gx_1 \ ,
\end{equation}
and the higher order Poisson tensors
\bd
	\Gp_i = {\mathcal R}^{i-1} \Gp_1 \ .
\ed

For a non-degenerate bi-Hamiltonian system, master symmetries can be generated using a method due to W. Oevel \cite{oevel}.

\begin{theorem}[Oevel]
Suppose that $X_0$ is a conformal symmetry for both $\Gp_1$, $\Gp_2$ and $H_1$,
	i.e., for some scalars $\Gl$, $\Gm$, and $\Gn$ we have
\bd
	{\mathcal L}_{X_0} \Gp_1= \Gl \Gp_1, \quad \quad
	{\mathcal L}_{X_0} \Gp_2 = \Gm \Gp_2, \quad \quad
	{\mathcal L}_{X_0} H_1 = \Gn H_1 \ .
\ed
Then the vector fields
\bd
	X_i = {\mathcal R}^i X_0
\ed
are master symmetries and we have

\renewcommand{\labelenumi}{(\alph{enumi})}
\begin{enumerate}
\item \makebox[\linewidth][c]{${\mathcal L}_{X_i} H_j = (\Gn +(j-1+i) (\Gm -\Gl)) H_{i+j}$\ ,}
\item \makebox[\linewidth][c]{${\mathcal L}_{X_i} \Gp_j = (\Gm +(j-i-2) (\Gm -\Gl)) \Gp_{i+j}$\ ,}
\item \makebox[\linewidth][c]{$[X_i, X_j]= (\Gm - \Gl) (j-i) X_{i+j}$\ .}
\end{enumerate}
\renewcommand{\labelenumi}{\arabic{enumi}.}
\end{theorem}

%------------------------------------------------------------------------------%
\section{The real Lie algebra \so{p}{q}}
%------------------------------------------------------------------------------%
In this section we construct a Toda type system associated with the real algebra $\so{p}{q}$.  The construction is similar in spirit to the one of Bogoyavlensky for complex simple Lie algebras.

We start with the real Lie algebra \so{p}{q} where $p=2m$ and $q=2n+1$ for some positive integers $m,n$.

To simplify the notation we let $N = m + n$.
Recall that by definition (see \cite{knapp})
	$$\so{p}{q} = \Set{X\in \gln[p+q]{\nR}}{X^{t}I_{p,q}+I_{p,q}X=0}$$
where $X^t$ is the transpose of $X$ and
$$
I_{p,q} =
\begin{pmatrix}
	I_p		& 0	\\
   0		& -I_q
\end{pmatrix}.
$$

We represent the elements of \so{p}{q} by real $(p+q)$-by-$(p+q)$ matrices
	of the form 	$\begin{pmatrix}a & b \cr b^t & d\end{pmatrix}$ 	with $a$ and $d$ skew symmetric, of size $p$-by-$p$ and $q$-by-$q$ respectively.

A basis for \so{p}{q} consists of the matrices $e_{ij}$ for $1 \leqslant i < j \leqslant p+q$ where:
\begin{alist}
\item For $1 \leqslant i < j \leqslant p$ and $p+1 \leqslant i < j \leqslant p+q$
	\be
	\textrm{the $kl$-th entry of } e_{ij} =
	\left\{
	\begin{array}{rll}
		1& \textrm{if} & (k,l) = (i,j)\\
		-1& \textrm{if} & (k,l) = (j,i)\\
		0& \textrm{if} & (k,l) \neq (i,j)
	\end{array}
	\right.
	\ee
\item For $1 \leqslant i \leqslant p$ and $p+1 \leqslant j \leqslant p+q$
	\be
	\textrm{the $kl$-th entry of } e_{ij} =
	\left\{
	\begin{array}{rll}
		1& \textrm{if} & (k,l) = (i,j)\\
		1& \textrm{if} & (k,l) = (j,i)\\
		0& \textrm{if} & (k,l) \neq (i,j)
	\end{array}
	\right.
	\ee
\end{alist}

A basis of the Cartan subalgebra \csH\ consists of $N$ block-diagonal matrices of size $2N+1$, 	whose first $N$ blocks are of size 2-by-2,	with \mbox{$i$-th} block $J = \begin{pmatrix}0 & 1 \cr -1 & 0\end{pmatrix}$, and
	whose last block is of size 1-by-1.
Note that we can express the elements of \csH\ in Kronecker product form if we drop their last row and column which consist entirely of 0's.
Hence, let
	$d_{i,i} = diag(0,\dots,1,\dots,0)$	be of size $N$-by-$N$ with 1 in the $i$-th position
	and let $\otimes$ denote the Kronecker product of matrices.
Then the $i$-th basis element of \csH\
	is given by	$d_{i,i} \otimes J$ and
	is equal to $e_{2i-1,2i}$ with the last row and column removed.
	
We denote the basis elements of \csH\ by $h_{\Ga_k}$ where the $\Ga_k$, are the simple roots defined as follows. For $1 \leqslant k \leqslant N-1$:
	\be
	\Ga_k(h_{\Ga_j}) =
	\left\{
	\begin{array}{rll}
		(-1)^k\; \ii& & j = k,\\
		(-1)^k\; \ii& & j = k+1,\\
		    0& & \textrm{otherwise}.
	\end{array}
	\right.
	\ee
Alternatively the simple roots correspond to vectors of size $N$ as follows:
\be
\begin{array}{rcl}
	\Ga_{1} &=& (-\ii,-\ii,0, \dots, 0),\\
	\Ga_{2} &=& (0,\ii,\ii,0, \dots, 0),\\
	\Ga_{3} &=& (0,0,-\ii,-\ii,0, \dots, 0),\\
	    &\vdots&\\
	\Ga_{N-1} &=& (0,0, \dots, (-1)^{N-1}\ii,(-1)^{N-1}\ii).
%	\Ga_N &=& (\ii,0, \dots, 0,(-1)^N\ii).
\end{array}
\ee
The root vectors corresponding to the roots $\Ga_k, -\Ga_{k}$, for $1 \leqslant k \leqslant N-1$, are given by:
\be
\begin{array}{rcl}
	x_{\Ga_k} &=& 1/2\left(
	e_{{2k-1,2k+1}}-e_{{2k,2k+2}} +
	(-1)^{k}(e_{{2k-1,2k+2}}+e_{{2k,2k+1}}) \ii
	\right)\\
	x_{-\Ga_k} &=& 1/2\left(
	e_{{2k-1,2k+1}}-e_{{2k,2k+2}} +
	(-1)^{k+1}(e_{{2k-1,2k+2}}+e_{{2k,2k+1}}) \ii
	\right).
\end{array}
\ee
We can express the root space in Kronecker product form as follows.
Firstly, let
\settowidth{\lengthX}{$\begin{pmatrix}1 & (-1)^k \ii \cr (-1)^k \ii & -1\end{pmatrix}$}
\be
	J_k = \begin{pmatrix}1 & (-1)^k\ \ii \cr (-1)^k\ \ii & -1\end{pmatrix}
	\quad \textrm{for $1 \leqslant k \leqslant N-1$.}
\ee
Secondly, for $1 \leqslant k \leqslant N-1$, let $X_k$ be an $N$-by-$N$ matrix defined as follows:\\
For $k \neq m$:
\be
	\textrm{the $ij$-th entry of } X_{k} =
	\left\{
	\begin{array}{rl}
		1/2& \textrm{if } (i,j) = (k,k+1)\\
		-1/2& \textrm{if } (i,j) = (k+1,k)\\
		0&  \textrm{otherwise.}
	\end{array}
	\right.
\ee
For $k = m$:
	\be
	\textrm{the $ij$-th entry of } X_{m} =
	\left\{
	\begin{array}{rl}
		1/2& \textrm{if } (i,j) = (m,m+1)\\
		1/2& \textrm{if } (i,j) = (m+1,m)\\
		0&  \textrm{otherwise.}
	\end{array}
	\right.
	\ee
Finally, for $1 \leqslant k \leqslant N-1$, we have:
	$$x_{\Ga_k} = X_k \otimes J_k.$$

One can define the system in Lax pair form ($L(t), B(t)$) in ${\mathcal G}$,
using the root system as follows:
\be
\begin{array}{lcl}\label{laxpair}
  L(t)&=&\text{$\displaystyle\sum_{j=1}^{N} (-1)^j b_j(t) h_{\Ga_j} +
  			\sum_{j=1}^{N-1} (-1)^j a_j(t) (x_{\Ga_j}+x_{-\Ga_j})$} \ ,\\
  B(t)&=&\text{$\displaystyle-\ii \sum_{j=1}^{N-1} (-1)^j a_j(t) (x_{\Ga_j}-x_{-\Ga_j})$} \ .
\end{array}
\ee
%\be \label{laxpair}
%	L(t)=\sum_{j=1}^{N} (-1)^j b_j(t) h_{\Ga_j} +
%		 \sum_{j=1}^{N-1} (-1)^j a_j(t) (x_{\Ga_j}+x_{-\Ga_j}) \ ,
%\ee
%\bd
%B(t)= -\ii \sum_{j=1}^{N-1} (-1)^j a_j(t) (x_{\Ga_j}-x_{-\Ga_j}) \ .
%\ed
As usual $h_{\alpha_j}$ is an element of a fixed Cartan subalgebra and $x_{\alpha_j}$ is a root vector corresponding to the simple root $\alpha_j$.
It is straightforward to verify that the Lax equation $\dot L(t) = [B(t), L(t)]$ is consistent and gives the following equations:
\be \label{dot}
\begin{array}{lcl}
 \dot a_i	&=& a_i \, (b_{i+1} -b_i ), \quad i = 1, 2, \dots , N-1\\
 \dot b_i	&=& 2 \, ( a_i^2 - a_{i-1}^2 ), \quad 1 \leqslant i \leqslant N \textrm{ and } i \nin \{1,m,m+1,N\},\\[2ex]
 \dot b_1	&=& 2 \, a_1^2\\
 \dot b_m	&=& -2 \, ( a_m^2 + a_{m-1}^2 )\\
 \dot b_{m+1} &=& 2 \, ( a_{m+1}^2 + a_{m}^2 )\\
 \dot b_N	&=& -2 \, a_{N-1}^2 \ .
\end{array}
\ee

Before giving an explicit matrix form of the Lax pair we remark that the last row and column of the matrices $L$ and $B$  consist entirely of zeros. For the sake of simplicity we will omit the last row and column of $L$ and $B$ since we still have a consistent Lax pair.

The Lax pair  (\ref{laxpair}),  can be written explicitly (after omitting the last row and column)  using the Kronecker product as follows:
\be \label{new L}
 L= L_1 \otimes
  \begin{pmatrix}
  0 & -1 \cr
  1 & 0
  \end{pmatrix}
  +
  L_2 \otimes
  \begin{pmatrix}
  1 & 0 \cr
  0 & -1
  \end{pmatrix}
\ee
where
\be \label{new L1}
L_1=
\begin{pmatrix}
	b_1  & 0   & \cdots & \cdots & 0   \cr
	0   & -b_2  & \ddots &    & \vdots \cr
	\vdots & \ddots & \ddots & \ddots & \vdots \cr
	\vdots &    & \ddots & \ddots & 0   \cr
   0	  & \cdots & \cdots & 0   & (-1)^{N+1}b_N
\end{pmatrix} \ ,
\ee

\be \label{new L2}
	L_2 = \left(
\begin{array}{ccccccclc}
	0		& -a_1	& 0		& \cdots 		&&		& \cdots 		&\:\:\:0\\
	a_1		& 0		& a_2	& \ddots 		&&		& 				&\:\:\:\;\! \vdots	\\
	0		& -a_2	& \ddots		& \ddots		&&		& 				&			\\	
	\vdots	& \ddots&\ddots & 		&&		&				& 			\\
      &    &    & (-1)^{m}a_m	&&  	&   			&			\\
			&		& 		& \phantom{xxx}\ddots	&	&			&\ddots	&\:\:\:\;\!\vdots	\\
			&		& 		& 		&&\ddots&\ddots	&\:\:\: 0	\\
	\vdots	&		& 		& 		&\ddots\phantom{x}&\ddots	& 0	&(-1)^{N-1}a_{N-1}		\\[2ex]
   0		&\cdots &		& 		&\cdots& 0	&(-1)^{N}a_{N-1}&\:\:\: 0		
\end{array} \right)
\ .
\ee
Note that the upper diagonal of $L_2$ is
	$$(-a_1, a_2, \dots, (-1)^{i}a_{i}, \dots, (-1)^{N-1}a_{N-1}),$$
whereas the lower diagonal is
	$$(a_1, -a_2, \dots, (-1)^{m}a_{m-1}, (-1)^{m}a_{m}, (-1)^{m+2}a_{m+1}, \dots , (-1)^{N}a_{N-1}).$$
That is
\be
	\textrm{the $ij$-th entry of } L_2 =
	\left\{
	\begin{array}{rl}
		(-1)^{i}a_{i}	& \textrm{if } j = i+1,\\
		(-1)^{j+1}a_{j}	& \textrm{if } j = i-1, j \neq m,\\
		(-1)^{j}a_{j}	& \textrm{if } j = i-1, j = m,\\
		0&  \textrm{otherwise.}
	\end{array}
	\right.
\ee
The matrix $B$ is defined as follows:
\be \label{new B}
 B= B_1 \otimes
  \begin{pmatrix}
  0 & 1 \cr
  1 & 0
  \end{pmatrix}
\ee
where
\be \label{new B1}
B_1 =
\begin{pmatrix}
  0     & a_1  & 0   & \cdots &    & \cdots  & 0    \cr
  -a_1   & 0   & a_2  & \ddots &    &     & \vdots \cr
  0     & -a_2  & 0   & \ddots &    &     &     \cr
  \vdots  & \ddots & \ddots & \ddots &    & \ddots  & \vdots \cr
       &    &    & a_m  &    & \ddots  & 0    \cr
  \vdots  &    &    & \ddots & \ddots & \ddots  & a_{N-1} \cr
  0     & \cdots &    & \cdots & 0   & -a_{N-1} & 0
\end{pmatrix}
\ .
\ee
Note that the upper diagonal of $B_1$ is $(a_1, a_2, \dots , a_{N-1})$,
whereas the lower diagonal is
$$(-a_1, \dots , -a_{m-1}, a_{m}, -a_{m+1}, \dots , -a_{N-1}).$$

It turns out that the trace of odd powers of $L$ are all zero.
\begin{lemma}
\bd
H_{2i+1}=\text{$\displaystyle  \tr L^{2i+1}$} = 0 \ \ \   \  \forall i  \ . \ed
\end{lemma}

\begin{proof}
Recall that by definition
	$$\so{p}{q} = \Set{X\in \gln[p+q]{\nR}}{X^{t}I_{p,q}+I_{p,q}X=0}$$
where $X^t$ is the transpose of $X$ and
$$
I_{p,q} =
\begin{pmatrix}
	I_p		& 0	\\
   0		& -I_q
\end{pmatrix}.
$$
It follows that $X=-I_{p,q}^{-1}X^t I_{p,q}$.
Let $f_n(x)$ be the characteristic polynomial of $X$, i.e.,
\bd f_n(\lambda)={\rm det}\ (\lambda I -X )  \ . \ed
Then
\begin{eqnarray*}
f_n( -\lambda)&=& {\rm det} (-\lambda I -X) = (-1)^n {\rm det} (\lambda I +X)\\
	&=&(-1)^n {\rm det} (\lambda I - I_{p,q}^{-1} X^t I_{pq} ) =(-1)^n  {\rm det}( I_{p,q}^{-1} I_{p,q} - I_{p,q}^{-1} X^t I_{pq})\\
	&=&(-1)^n {\rm det}( I_{p,q}^{-1} (\lambda I- X^t) I_{p,q} ) =(-1)^n {\rm det} (\lambda I- X^t)=(-1)^n \end{eqnarray*}
%\bd f_n( -\lambda)= {\rm det} (-\lambda I -X) = (-1)^n {\rm det} (\lambda I +X) \ed
%
%\bd =(-1)^n {\rm det} (\lambda I - I_{p,q}^{-1} X^t I_{pq} ) =(-1)^n  {\rm det}( I_{p,q}^{-1} I_{p,q} - I_{p,q}^{-1} X^t I_{pq}) \ed
%
%\bd =(-1)^n {\rm det}( I_{p,q}^{-1} (\lambda I- X^t) I_{p,q} ) =(-1)^n {\rm det} (\lambda I- X^t)=(-1)^n f_n(\lambda)  \ . \ed
Since the Lax matrix is of odd dimension, $n$ is odd. Dividing by $\lambda$ (which corresponds to eigenvalue zero) we obtain an even polynomial.
Therefore, for the Lax matrix $L$  if $\lambda$ is an eigenvalue of $L$ so is $-\lambda$. This implies that characteristic polynomial of $L$  is an even polynomial. It also shows that ${\rm tr} \ {L}^k=0$ if $k$ is odd.
\end{proof}

As we will see shortly the function $H_1=b_1+b_2+ \dots + b_{N-1}$  (which is not the trace of $L$)  will turn out to be a Casimir.  Therefore we need $N-1$ independent functions in involution is order to establish integrability.

The following functions are constants of motion:
\bd
H_{2i}  = \text{$\displaystyle \frac{(-1)^i}{4i}\, \tr  L^{2i}$},	\textrm{ for } i = 1, \dots N-1.
\ed

%------------------------------------------------------------------------------%
\section{Symplectic realization}
%------------------------------------------------------------------------------%
In ${\nR}^{2N}$ with coordinates $({\bf q}, {\bf p})$, 	${\bf q}=(q_1, \dots, q_N)$, and
	${\bf p}=(p_1, \dots, p_N)$ we choose the following Hamiltonian function:
\be \label{a2np}
%	H(q_1, \dots, q_N, \, p_1, \dots, p_N) =
	H(\mathbf{q},\mathbf{p}) =
		\sum_{i=1}^N \, \frac{1}{2} \, p_i^2 +
		\left(\sum _{i=1, i \neq m}^{N-1} \, e^{ q_i-q_{i+1}}\right) - e^{ q_m-q_{m+1}} \ ,
\ee
which differs from the classical, non-periodic Toda lattice by a minus sign at position $m$.

Using Flaschka's  transformation $F: ({\bf q}, {\bf p}) \to ({\bf a,b})$  we  obtain equations (\ref{dot}).  Recall that the  $(a,b)$ variables are defined by:
\be
\begin{array}{lcll}
	a_i & = &\frac{1}{2}e^{ \frac{1}{2} (q_i-q_{i+1})},	\quad &i=1,2, \dots, N-1, \\
	b_i & = &-\frac{1}{2} p_i,							\quad &i=1,2, \dots, N.
	\label{FlaschkaTransformation}
\end{array}
\ee

Let the mapping $F: {\bf R}^{2N} \to {\bf R}^{2N-1}$, where
\bd
 F: (q_1, \dots, q_N, p_1, \dots, p_N) \to (a_1, \dots, a_{N-1}, b_1, \dots, b_N) \ ,
\ed
denote the Flaschka transformation given by equation (\ref{f1}).
Recall that the standard symplectic bracket on ${\bf R}^{2N}$ is mapped onto the linear bracket $\pi_1$ in $({\bf a,b})$ coordinates. The linear bracket has only one Casimir 	$H_1=b_1+b_2 + \dots +b_N$.

We can take as Hamiltonian the function $H_2 = -\frac{1}{4}\, \tr\, L^2$,
	where $L$ is given by (\ref{new L}).

We define a quadratic bracket which we call  $\Gp_2$ with defining relations:
\be \label{poisson p2}
\begin{array}{rcll}
	\{a_i, a_{i+1} \}	&=& \frac{1}{2} a_i a_{i+1} \ , \quad & i= 1, \dots , N-1,\\
	\{a_i, b_i \}		&=& -a_i b_i \ , \quad & i= 1, \dots , N-1,\\
	\{a_i, b_{i+1} \}	&=& a_i b_{i+1} \ , \quad & i= 1, \dots , N-1, \\
	\{b_i, b_{i+1} \}	&=& 2\, a_i^2 \ , \quad & i= 1, \dots , m-1, m+1 , \dots, N-1, \\
	\{b_m, b_{m+1} \}	&=&  -2\, a_m^2 \ ;
\end{array}
\ee
all other brackets are zero. Note that this bracket is different from the Adler bracket defined in (\ref{Adler_bracket}).
We remark that the mapping of ${\bf R}^{2N-1} \to {\bf R}^{2N-1}$ given by
\bd (a_1, a_2, \dots , a_{N-1}, b_1, b_2, \dots, b_n) \to (a_1, \dots, a_{m-1}, i a_m, a_{m+1}, \dots,  a_{N-1}, b_1, b_2, \dots, b_n) \ed
sends the Adler bracket $\Pi_2$ onto the bracket $\pi_2$. Since $\Pi_2$ is known to be Poisson this implies that $\pi_2$ is also Poisson. We simply think of the two brackets as complex brackets connected by a complex Poisson map and therefore the Jacobi identity holds.  Since the coefficients of the defining functions of $\pi_2$ are real it is also Poisson as a real Poisson manifold.

This bracket also has
	one Casimir ${\rm det} \, L$ and
	$H_1 = b_1 + \dots b_N$ is the Hamiltonian.

Furthermore, $\pi_2$ is compatible with $\pi_1$.
We also have
\be
\pi_1 \nabla H_2 = \pi_2 \nabla H_1 \ . \label{a98}
\ee

Similarly we define another Poisson bracket, $\pi_3$, which is a modification of the cubic Toda Poisson bracket. The defining relations for $\pi_3$ are
\be \label{poisson p3}
\begin{array}{rcll}
	\{a_i, a_{i+1} \}	&=& a_i a_{i+1} b_{i+1} \ , \quad & i= 1, \dots , N-2,\\
	\{a_i, b_i \}		&=& -a_i b_i^2-a_i^3 \ , \quad & i= 1, \dots , m-1, m+1 , \dots, N-1,\\
	\{a_i, b_i \}		&=& -a_i b_i^2+a_i^3 \ , \quad & i=m,\\
	\{a_i, b_{i+1}\}	&=& a_i b_{i+1}^2 +a_i^3 \ , \quad & i= 1, \dots , m-1, m+1 , \dots, N-1,\\
	\{a_i, b_{i+1}\}	&=& a_i b_{i+1}^2 -a_i^3 \ , \quad & i=m,\\
	\{a_i, b_{i+2} \}	&=& a_i a_{i+1}^2 \ , \quad & i= 1, \dots , m-2, m , \dots, N-2,\\
	\{a_i, b_{i+2} \}	&=& -a_i a_{i+1}^2 \ , \quad & i=m-1,\\
	\{a_{i+1}, b_i \}	&=& -a_i^2 a_{i+1} \ , \quad & i= 1, \dots , m-1, m+1 , \dots, N-2,\\
	\{a_{i+1}, b_i \}	&=& a_i^2 a_{i+1} \ , \quad & i=m,\\
	\{b_i, b_{i+1} \}	&=& 2\, a_i^2 \, (b_i+b_{i+1}) \ , \quad & i= 1, \dots , m-1, m+1 , \dots, N-1,\\
	\{b_i, b_{i+1} \}	&=& -2\, a_i^2 \, (b_i+b_{i+1}) \ , \quad & i=m;
\end{array}
\ee
all other brackets are zero.

The bracket $\pi_3$ is compatible with both $\pi_1$ and $\pi_2$  and satisfy Lenard type relations like
\bd \pi_1 dH_4= \pi_3 dH_2  \ . \ed

We note that
\be \label{H1}
H_1=\sum_{i=1}^N b_i=-\frac{1}{2} (p_1+p_2+ \dots + p_N) \ ,
\ee
corresponds to the total momentum ($H_1 \neq \tr L$) and
\be \label{H2}
H_2	= H(q_1, \dots, q_N, \, p_1, \dots, p_N)
	= \frac{1}{2}\sum_{i=1}^N b_i^2+ \left(\sum_{i=1, i \neq m}^{N-1} a_i^2 \right)-a_m^2
\ee
is the Hamiltonian.

We  now  define Toda hierarchies for the Toda lattice in $(q,p)$ variables.
We follow reference \cite{damianou2}.

Let $\hat{J}_1$ be the standard  symplectic bracket  with Poisson matrix
\bd
\hat{J}_1 =
\begin{pmatrix}
	0 & I \cr
	-I & 0
\end{pmatrix} \ ,
\ed
where $I$ is the $N \times N$ identity matrix. We use $J_1=4 \hat{J}_1$.
With this convention the bracket $J_1$ is mapped precisely onto the bracket $\pi_1$ under the Flaschka transformation (\ref{f1}).
We define $\hat{J}_2$ to be the tensor
\bd
\hat{J}_2 =
\begin{pmatrix}
	A & B \cr
	-B & C
\end{pmatrix} \ ,
\ed
where
$A$ is the skew-symmetric matrix defined by $a_{ij}=1=-a_{ji}$ for $i<j$,
$B$ is the diagonal matrix $(-p_1, -p_2, \dots, -p_N)$ and
$C$ is the skew-symmetric matrix whose non-zero terms are
	$c_{i,i+1}=-c_{i+1,i}=e^{q_i-q_{i+1}}$ for $i=1,2, \dots,m-1,m+1,\dots, N-1$ and
	$c_{m,m+1}=-c_{m+1,m}=-e^{q_m-q_{m+1}}$.
We define $J_2=2 \hat{J}_2$.
 With this convention the
bracket $J_2$ is mapped precisely onto the bracket $\pi_2$ under the Flaschka transformation.

It is easy to see that we have a bi-Hamiltonian pair. We define
\bd
h_1=-2(p_1+p_2+\dots +p_N) \ ,
\ed
and $h_2$ to be the Hamiltonian:
\bd
h_2=
	\sum_{i=1}^N \, \frac{1}{2} \, p_i^2 +
	\left(\sum _{i=1, i \neq m}^{N-1} \, e^{ q_i-q_{i+1}}\right) - e^{ q_m-q_{m+1}} \ .
\ed

Under Flaschka's transformation (\ref{f1}),
	$h_1$ is mapped onto $4(b_1+b_2+ \dots+b_N) = 4 H_1$ and
	$h_2$ is mapped onto $-\tr  L^2 = 4 H_2$.
Using the relationship (\ref{a98}), we obtain, after multiplication by 4, the
following pair:
\bd
J_1 \nabla h_2= J_2 \nabla h_1 \ .
\ed
We define the recursion operator as follows:
\bd
{\mathcal R}=J_2 J_1^{-1} \ .
\ed
The matrix form of ${\mathcal R}$ is quite simple:
\be
{\mathcal R} =\frac{1}{2}
\begin{pmatrix}
	B &-A \cr
	C& B
\end{pmatrix}
\ . \label{c1}
\ee

In $(q,p)$ coordinates, the symbol $\chi_i$ is a shorthand for $\chi_{h_i}$. It is generated as
usual by
\bd
\chi_i = {\mathcal R}^{i-1} \chi_1 \ .
\ed
In a similar fashion we obtain the higher order Poisson tensors
\bd
J_i = {\mathcal R}^{i-1} J_1 \ .
\ed
We finally define the conformal symmetry
\bd
Z_0=\sum_{i=1}^N (N-2i+1) \frac{\partial}{\partial q_i} +\sum_{i=1}^N p_i \frac{\partial}{\partial p_i} \ .
\ed
It is straightforward to verify that
\begin{eqnarray*}
{\mathcal L}_{Z_0} J_1&=&- J_1 \ ,\\
{\mathcal L}_{Z_0} J_2&=&0 \ .
\end{eqnarray*}
In addition,
\begin{eqnarray*}
Z_0(h_1)&=&h_1\\
Z_0(h_2)&=&2h_2 \ .
\end{eqnarray*}

Consequently, $Z_0$ is a conformal symmetry for $J_1$, $J_2$ and $h_1$.
The constants appearing in Oevel's Theorem are
	$\lambda=-1$, $\mu=0$ and $\nu=1$.
Therefore, we end--up with the following deformation relations:
\begin{eqnarray*}
[Z_i, h_j]&=& (i+j)h_{i+j}\\
{\mathcal L}_{Z_i} J_j &=& (j-i-2) J_{i+j} \\ 
\text{$[Z_i, Z_j]$} &=& (j-i) Z_{i+j} \ .
\end{eqnarray*}

%\bd
%[Z_i, h_j]= (i+j)h_{i+j}
%\ed
%\bd
%{\mathcal L}_{Z_i} J_j = (j-i-2) J_{i+j}
%\ed
%\bd
% [ Z_i, Z_j ] = (j-i) Z_{i+j} \ .
%\ed

We note that Magri's theorem implies the Lenard type relations
 \be \label{lenardpq}
 J_i dh_j=J_j dh_i   \ .
\ee

It is straightforward to verify that the symplectic bracket $J_1$ is mapped via Flaschka's transformation $F$ onto the linear Poisson bracket $\pi_1$ and  $J_2$ is mapped onto the quadratic bracket $\pi_2$.  Similarly $h_1$, $h_2$ correspond via $F$ to $H_1$ and $H_2$ respectively. The master symmetries $Z_i$ are mapped via $F$ to produce a sequence of master symmetries $X_i$ in the $(a,b)$ phase space.

%------------------------------------------------------------------------------%
\section{Integrability}
%------------------------------------------------------------------------------%
In this section we prove the integrability of  the $so(p,q)$ Toda system.
We note that the tensors $J_1$ and $J_2$  project under Flaschka's transformation to the linear and quadratic brackets  $\pi_1$ and $\pi_2$ respectively. As in the case of the classical Toda we have the bi-hamiltonian formulation.
\begin{equation}
\pi_1 \nabla H_2 = \pi_{2} \nabla H_1   \ .
\end{equation}
Of course this bi-hamiltonian formulation is quite different than the Lenard relations (\ref{lenardpq}).  The Poisson tensors $\pi_1$ and $\pi_2$ are no longer symplectic and as we know there exists no recursion operator in the $(a,b)$ space. It is well-known (see  \cite{fernandes}, \cite{damianou2}, \cite{morosi}) that  the recursion operator ${\mathcal R}$ cannot be  reduced.  It is also well-known that the vector fields $Z_i$ are projectable and give rise to a sequence of vector fields $X_1, X_2, \dots $ in the $(a,b)$ space \cite{fernandes}. To prove the involution of the invariants we will only use the master symmetry $X_2$ and the reduced Hamiltonians $H_{2i}$ which are the constants of motion for our system.
We will make use of the Lenard relations
\be \label{lenardab}
\pi_1 dH_{2i}=\pi_3 dH_{2i-2}   \ \  \forall i  \ .
\ee

\begin{theorem}
The functions $H_{2i}$ are in involution with respect to the $\pi_1$ Poisson bracket.

\end{theorem}

\begin{proof}
Using the Lenard relations (\ref{lenardab})  it is straightforward to prove the involution of the functions $H_{2i}$.  The Poisson bracket we use $\{ , \}$ is the linear Poisson bracket $\pi_1$.

Since $H_2$ is the Hamiltonian  we have  $\{ H_2, H_{2j} \}=0$  $\forall j$   since $H_{2j}$ is clearly a constant of motion.
Now we  calculate
\bd
\begin{array} {lcl}
\{H_{2i}, H_{2j} \} &=& <dH_{2i},\  \pi_1 dH_{2j}>  \\
                &=&-<dH_{2j},\  \pi_1 dH_{2i}> \\
                &=&-<dH_{2j}, \  \pi_3 dH_{2i-2}> \\
                &=&<dH_{2i-2}, \  \pi_3 dH_{2j} > \\
                &=&<dH_{2i-2}, \  \pi_1 dH_{2j+2}>  \\
                &=&\{H_{2i-2}, H_{2j+2}  \} \\
                &&   \ \ \ \ \ \  \vdots \\
                &&    \ \ \ \ \ \ \vdots \\
                &=&   \{H_2, H_{2i+2j-2}  \}= 0 \ .
\end{array}
\ed

\end{proof}

It is not difficult to see that the functions $H_{2i}$, $i=1,2, \dots, N-1$ are functionally independent on an open dense set. Using the original approach of Henon \cite{henon} it is enough to assume $a_1=a_2=\dots =a_{N-1}=0$. Then
the traces of $L^{2i}$ are symmetric functions of the $b_i$.  In fact taking into account the form of the Lax matrix (\ref{new L} ) it is easy to see that

\bd H_{2i} =\frac{1}{2i}(b_1^{2i} + ... + b_N^{2i}) \ed
at ${\bf a}=0$.   Therefore the Jacobian matrix has the form

\bd
\begin{pmatrix}
	b_1  & b_2  & \cdots & \cdots & b_N   \cr
	b_1^3   & b_2^3  & \cdots &    & b_N^3 \cr
	b_1^5 & b_2^5 & \cdots & \cdots & b_N^5 \cr
	\vdots &    & \vdots & \vdots & \vdots   \cr
    b_1^{2N-3}  & \cdots & \cdots &   & b_N^{2N-3}
\end{pmatrix} \ .
\ed

Selecting  the top left corner $N-1 \times N-1$ minor we obtain a  Vandermonde matrix whose determinant is different than zero on the set $b_i \not=b_j$,  for  $ i\not=j$.

We can prove integrability in a different way, by finding a different Lax pair and a new set of independent functions in involution.  Define the following Lax pair $(A, M)$ where

\begin{equation}
M =
\begin{pmatrix}
	b_1 &  a_1 & 0 & \cdots & \cdots & 0 \cr
	a_1 & b_2 & a_2 & \cdots &    & \vdots \cr
	0 & a_2 & b_3 & \ddots &  &  \cr
	\vdots & & \ddots &  \ddots &\ii a_m & \vdots \cr
	\vdots & & &\ii a_m & \ddots & a_{N-1} \cr
	0 & \cdots & & \cdots & a_{N-1} & b_N
\end{pmatrix}
                        \ ,
\end{equation}
i.e. we replace $a_m$ with $\ii  a_m$ where  $\ii^2=-1$.  We take $A$ to be the skew-symmetric part of $M$.  One verifies easily that the Lax pair is equivalent to the equations (\ref{dot}).  Therefore the functions
\bd I_i={1 \over i} {\rm tr} \, M^i,  \   \ i=1,2, \dots, N  \ed
are first integrals of the system (\ref{dot}).  As we have seen $I_1=H_1=b_1 +b_2+ \dots +b_N$ is the Casimir of the Poisson bracket $\pi_1$. Therefore  we may use  the set of integrals $\{I_2, I_3, \dots I_N \}$ as a different set of functions to be used to prove integrability.  Note that  $H_{2i}$  is equal to $I_{2i}$ . To prove that this set is in involution one can use the methods of \cite{damianou1}.   Recall that the first master symmetry of the classical  Toda lattice is
\begin{equation}
  \sum_{n=1}^{N-1} \alpha_n {\partial \ \over \partial a_n} +
   \sum _{n=1}^N \beta_n  {\partial \ \over \partial b_n} \ ,
 \end{equation}
  \smallskip
  \noindent
where
\begin{equation}
 \alpha_n =  - na_n b_n + (n+2) a_n b_{n+1 }
\end{equation}
\begin{equation}
\beta_n = (2n+3) a_n^2 + (1-2n) a_{n-1}^2 + b_n^2  \ .
\end{equation}

Replacing $a_m$ by $\ii a_m$ in the first equation does not alter the equation.    On the other hand  $\beta_m $ now changes to
\bd \beta_m=-(2m+3) a_m^2 + (2m-1)  a_{m-1}^2 + b_m^2 \ . \ed

With this change we obtain master symmetry $X_1$ which has the following properties.
\bd X_1 (H_j) =(j+1) H_{j+1} \ed
\bd {\mathcal L}_{X_1} \pi_j =(j-3) \pi_{j+1}  \ j=1,2,3  \ . \ed

One now establishes the integrability of system (\ref{dot}) following  \cite{damianou1, damianou2}. The proof of involution of invariants is exactly the same as in   \cite[Proposition~3]{damianou2}.

%------------------------------------------------------------------------------%
\section{Example  \so{6}{5}}
%------------------------------------------------------------------------------%
We illustrate in detail  the results with a specific example.

Using the root structure of  \so{6}{5} we define the following Lax pair.

The matrix $L$ is given by
\[\displaystyle
L\, = \, \left(
\begin {array}{cccccccccc} 0&-b_{{1}}&-a_{{1}}&0&0&0&0&0&0&0\\\noalign{\medskip}b_{{1}}&0&0&a_{{1}}&0&0&0&0&0&0\\\noalign{\medskip}a_{{1}}&0&0&b_{{2}}&a_{{2}}&0&0&0&0&0\\\noalign{\medskip}0&-a_{{1}}&-b_{{2}}&0&0&-a_{{2}}&0&0&0&0\\\noalign{\medskip}0&0&-a_{{2}}&0&0&-b_{{3}}&-a_{{3}}&0&0&0\\\noalign{\medskip}0&0&0&a_{{2}}&b_{{3}}&0&0&a_{{3}}&0&0\\\noalign{\medskip}0&0&0&0&-a_{{3}}&0&0&b_{{4}}&a_{{4}}&0\\\noalign{\medskip}0&0&0&0&0&a_{{3}}&-b_{{4}}&0&0&-a_{{4}}\\\noalign{\medskip}0&0&0&0&0&0&-a_{{4}}&0&0&-b_{{5}}\\\noalign{\medskip}0&0&0&0&0&0&0&a_{{4}}&b_{{5}}&0
\end {array} \right)
\]

or equivalently in Kronecker product form
\begin{equation*}
\begin{split}
L \quad =&\quad
	\begin{pmatrix}
	0 & -a _{1} & 0 & 0 & 0 \\
	a _{1} & 0 & a _{2} & 0 & 0 \\
	0 & -a _{2} & 0 & -a _{3} & 0 \\
	0 & 0 & -a _{3} & 0 & a _{4} \\
	0 & 0 & 0 & -a _{4} & 0
	\end{pmatrix}
\otimes
	\begin{pmatrix}
	1 & 0 \\
	0 & -1
	\end{pmatrix}
\\
&+
	\begin{pmatrix}
	b _{1} & 0 & 0 & 0 & 0 \\
	0 & -b _{2} & 0 & 0 & 0 \\
	0 & 0 & b _{3} & 0 & 0 \\
	0 & 0 & 0 & -b _{4} & 0 \\
	0 & 0 & 0 & 0 & b _{5}
	\end{pmatrix}
\otimes
	\begin{pmatrix}
	0 & -1 \\
	1 & 0
	\end{pmatrix} \ .
\end{split}
\end{equation*}

The matrix $B$ is the following:

\[\displaystyle
B\, = \, \left(
\begin {array}{cccccccccc} 0&0&0&a_{{1}}&0&0&0&0&0&0\\\noalign{\medskip}0&0&a_{{1}}&0&0&0&0&0&0&0\\\noalign{\medskip}0&-a_{{1}}&0&0&0&a_{{2}}&0&0&0&0\\\noalign{\medskip}-a_{{1}}&0&0&0&a_{{2}}&0&0&0&0&0\\\noalign{\medskip}0&0&0&-a_{{2}}&0&0&0&a_{{3}}&0&0\\\noalign{\medskip}0&0&-a_{{2}}&0&0&0&a_{{3}}&0&0&0\\\noalign{\medskip}0&0&0&0&0&a_{{3}}&0&0&0&a_{{4}}\\\noalign{\medskip}0&0&0&0&a_{{3}}&0&0&0&a_{{4}}&0\\\noalign{\medskip}0&0&0&0&0&0&0&-a_{{4}}&0&0\\\noalign{\medskip}0&0&0&0&0&0&-a_{{4}}&0&0&0
\end {array} \right)
\]

or equivalently
\begin{equation*}
B \quad =\quad
	\begin{pmatrix}
	0 & a _{1} & 0 & 0 & 0 \\
	-a _{1} & 0 & a _{2} & 0 & 0 \\
	0 & -a _{2} & 0 & a _{3} & 0 \\
	0 & 0 & a_{3} & 0 & a _{4} \\
	0 & 0 & 0 & -a _{4} & 0
	\end{pmatrix}
\otimes
	\begin{pmatrix}
	0 & 1 \\
	1 & 0
	\end{pmatrix} \ .
\end{equation*}

The Lax pair is equivalent to the following equations of motion:
\[\displaystyle
\begin {array}{rclccrcl} \dot{b_1}(t)&=&2\,{a_{1}}^{2}&&&\dot{a_1}(t)&=& (b_{2}-b_{1}) a_{1}\\\noalign{\medskip}\dot{b_2}(t)&=&2\,{a_{2}}^{2}-2\,{a_{1}}^{2}&&&\dot{a_2}(t)&=& (b_{3}-b_{2}) a_{2}\\\noalign{\medskip}\dot{b_3}(t)&=&-2\,{a_{3}}^{2}-2\,{a_{2}}^{2}&&&\dot{a_3}(t)&=&(b_{4}-b_{3}) a_{3}\\\noalign{\medskip}\dot{b_4}(t)&=&2\,{a_{4}}^{2}+2\,{a_{3}}^{2}&&&\dot{a_4}(t)&=& (b_{5}-b_{4}) a_{4}\\\noalign{\medskip}\dot{b_5}(t)&=&-2\,{a_{4}}^{2}&&&&&\\
\end {array}\]

The constants of motion are 		$H_{2i}=\frac{(-1)^i}{4i} \tr L^{2i} \ \  i=1,2,3,4$ which are sufficient to establish integrability since $H_1 = \sum_{j=1}^5 b_j$ is a Casimir for the $\pi_1$ bracket.
	
For example:
\begin{eqnarray*}
H_1 &=&
{b_1}+{b_2}+{b_3}+{b_4}+{b_5}
\\
H_2 &=&
{a_1}^{2}+{a_2}^{2}-{a_3}^{2}+{a_4}^{2}
+\frac{1}{2}(b_{1}^{2}+{b_2}^{2}+{b_3}^{2}+{b_4}^{2}+{b_5}^{2})
\\
H_4 &=&
{a_1}^{2}({b_1}^{2}+b_1b_2+{b_2}^{2})
+{a_2}^{2}({b_2}^{2}+b_2b_3+{b_3}^{2})
-{a_3}^{2}({b_3}^{2}+b_3b_4+{b_4}^{2})
\\&&
+{a_4}^{2}({b_4}^{2}+b_4b_5+{b_5}^{2})
%\\&&
+{a_1}^{2}{a_2}^{2}
-{a_2}^{2}{a_3}^{2}
-{a_3}^{2}{a_4}^{2}
\\&&
+1/2({a_1}^{4}+{a_2}^{4}+{a_3}^{4}+{a_4}^{4})
%\\&&
+1/4({b_1}^{4}+{b_2}^{4}+{b_3}^{4}+{b_4}^{4}+{b_5}^{4})
\end{eqnarray*}

Note that $H_2, H_4, H_6, H_8$ form a set of independent integrals of motion in involution demonstrating the integrability of the system.

We obtain another set of integrals in involution in the following way: Define the Lax pair $(M, A)$ where

\bd
M=
\begin{pmatrix}
	b_1 &  a_1 & 0 & 0 &  0 \cr
	a_1 & b_2 & a_2 & 0    &0 \cr
	0 & a_2 & b_3 & \ii a_3 &  0  \cr
	0 &0 &\ii a_3 & b_4&  a_4 \cr
	0 & 0 & 0 & a_4 & b_5
\end{pmatrix}
\ed
and

\bd
A=
\begin{pmatrix}
	0 &  a_1 & 0 & 0 &  0 \cr
	-a_1 &0 & a_2 & 0    &0  \cr
	0 &- a_2 &0 & \ii a_3 &  0  \cr
	0 &0 &-\ii a_3 & 0&  a_4 \cr
	0 & 0 & 0 &- a_4 &0
\end{pmatrix} \ .
\ed

Define: $I_1 =H_1$, $I_j = \frac{1}{j} X_1(I_{j-1})$

We have $X_1(H_1)=2H_2$. Define
\[ I_2 =\frac{1}{2} X_1(H_1)=H_2 \ . \]

Then define
\[ I_3=\frac{1}{3} X_1 (H_2)  \ , \]
and
\[ I_4=\frac{1}{4} X_1 (I_3) =H_4 \ . \]

Finally let
\[ I_5=\frac{1}{5} X_1 (I_4) \ . \]

Then the set $\{ I_2=H_2, I_3, I_4=H_4, I_5 \}$ is another set of independent integrals in involution.

Of course
\bd I_j=\frac{1}{j}\, {\rm tr}\  M^j  \ . \ed
For example
\bd I_3=\frac{1}{3} \sum_{j=1}^5 b_j^3+a_1^2(b_1+b_2)+a_2^2(b_2+b_3)-a_3^2 (b_3+b_4) +a_4^2 (b_4+b_5) \ . \ed

In addition we have $\{ I_j, H_k \}=0$ for $j, k=1,2,3,4$. Therefore the functions $H_6$ and $H_8$ must be functions of the $I_j$. For example
\[ H_6=\frac{1}{720} H_1^6 +H_1 I_5+H_2 H_4-\frac{1}{2} H_1^2 H_4+ \frac{1}{6} H_1^3 I_3-H_1 H_2 I_3-\frac{1}{24} H_1^4 H_2 -\frac{1}{6} H_2^3+\frac{1}{2}I_3^2+\frac{1}{4}H_1^2 H_2^2 \ . \]

The equations of motion can be also obtained using the
Lie-Poisson bracket $\pi_1$

\[\displaystyle
{\it \Gp_1}= \left(
\begin {array}{cccccccccc}
0&0&0&0&-a_{{1}}&a_{{1}}&0&0&0\\\noalign{\medskip}
0&0&0&0&0&-a_{{2}}&a_{{2}}&0&0\\\noalign{\medskip}
0&0&0&0&0&0&-a_{{3}}&a_{{3}}&0\\\noalign{\medskip}
0&0&0&0&0&0&0&-a_{{4}}&a_{{4}}\\\noalign{\medskip}
a_{{1}}&0&0&0&0&0&0&0&0\\\noalign{\medskip}
-a_{{1}}&a_{{2}}&0&0&0&0&0&0&0\\\noalign{\medskip}
0&-a_{{2}}&a_{{3}}&0&0&0&0&0&0\\\noalign{\medskip}
0&0&-a_{{3}}&a_{{4}}&0&0&0&0&0\\\noalign{\medskip}
0&0&0&-a_{{4}}&0&0&0&0&0
\end {array} \right)
\]
and the Hamiltonian $H_2$.

There is one Casimir for $\Gp_1$:
	$$H_1=b_1+b_2 + \dots +b_5.$$

The quadratic  Poisson bracket is defined by:
\[\displaystyle {\it \Gp_2}=
\left(
\begin {array}{cccccccccc}
0&\frac{a_{1}a_{2}}{2}&0&0&-a_{1}b_{1}&a_{1}b_{2}&0&0&0\\\noalign{\medskip}
-\frac{a_{1}a_{2}}{2}&0&\frac{a_{2}a_{3}}{2}&0&0&-a_{2}b_{2}&a_{2}b_{3}&0&0\\\noalign{\medskip}
0&-\frac{a_{2}a_{3}}{2}&0&\frac{a_{3}a_{4}}{2}&0&0&-a_{3}b_{3}&a_{3}b_{4}&0\\\noalign{\medskip}
0&0&-\frac{a_{3}a_{4}}{2}&0&0&0&0&-a_{4}b_{4}&a_{4}b_{5}\\\noalign{\medskip}
a_{1}b_{1}&0&0&0&0&2\,{a_{1}}^{2}&0&0&0\\\noalign{\medskip}
-a_{1}b_{2}&a_{2}b_{2}&0&0&-2\,{a_{1}}^{2}&0&2\,{a_{2}}^{2}&0&0\\\noalign{\medskip}
0&-a_{2}b_{3}&a_{3}b_{3}&0&0&-2\,{a_{2}}^{2}&0&-2\,{a_{3}}^{2}&0\\\noalign{\medskip}
0&0&-a_{3}b_{4}&a_{4}b_{4}&0&0&2\,{a_{3}}^{2}&0&2\,{a_{4}}^{2}\\\noalign{\medskip}
0&0&0&-a_{4}b_{5}&0&0&0&-2\,{a_{4}}^{2}&0
\end {array}
\right)
\]

 There is one Casimir for $\Gp_2$:
\begin{eqnarray*}
\sqrt{{\rm det} \, L} &=&
- b_{1} b_{2} b_{3} b_{4} b_{5}
+ b_{1} b_{2} b_{3} a_{4}^{2}
- b_{1} b_{2} b_{5} a_{3}^{2}
+ b_{1} b_{4} b_{5} a_{2}^{2}
- b_{1} a_{2}^{2} a_{4}^{2}
\\&&
+ b_{3} b_{4} b_{5} a_{1}^{2}
- b_{3} a_{1}^{2} a_{4}^{2}
+ b_{5} a_{1}^{2} a_{3}^{2}=-{\rm det} M
\end{eqnarray*}

The formulas for the master symmetries $X_1$ and $X_2$ are the following:

\[\displaystyle
X_{1} =
\left(
\begin{array}{c}
\frac{1}{2}\,a_{1} \left( 5\,b_{1}-b_{2} \right) \\\noalign{\medskip}
\frac{1}{2}\,a_{2} \left( 3\,b_{2}+b_{3} \right) \\\noalign{\medskip}
\frac{1}{2}\,a_{3} \left( 3\,b_{4}+b_{3} \right) \\\noalign{\medskip}
-\frac{1}{2}\,a_{4} \left( b_{4}-5\,b_{5} \right) \\\noalign{\medskip}
-2\,{a_{1}}^{2}+{b_{1}}^{2}\\\noalign{\medskip}
4\,{a_{1}}^{2}+{b_{2}}^{2}\\\noalign{\medskip}
2\,{a_{2}}^{2}-2\,{a_{3}}^{2}+{b_{3}}^{2}\\\noalign{\medskip}
4\,{a_{4}}^{2}+{b_{4}}^{2}\\\noalign{\medskip}
-2\,{a_{4}}^{2}+{b_{5}}^{2}
\end{array}
\right)
\]

\[\displaystyle
X_{2} =
\left(
\begin {array}{c}
\frac{1}{2}\,a_{1} \left( 2\,b_{1}b_{2}+b_{4}b_{1}+b_{5}b_{1}-b_{3}b_{2}-b_{4}b_{2}-b_{5}b_{2}+b_{3}b_{1}+5\,{b_{1}}^{2}-{b_{2}}^{2}+2\,{a_{1}}^{2} \right) \\\noalign{\medskip}
\frac{1}{2}\,a_{2} \left( 3\,{b_{2}}^{2}-b_{1}b_{2}+2\,b_{3}b_{2}+b_{4}b_{2}+b_{5}b_{2}+4\,{a_{1}}^{2}+2\,{a_{2}}^{2}-2\,{a_{3}}^{2}+{b_{3}}^{2}+b_{3}b_{1}-b_{4}b_{3}-b_{5}b_{3} \right) \\\noalign{\medskip}
\frac{1}{2}\,a_{3} \left( -b_{3}b_{1}-b_{3}b_{2}+2\,b_{4}b_{3}+b_{5}b_{3}+4\,{a_{4}}^{2}+3\,{b_{4}}^{2}+b_{4}b_{1}+b_{4}b_{2}-b_{5}b_{4}+2\,{a_{2}}^{2}-2\,{a_{3}}^{2}+{b_{3}}^{2} \right) \\\noalign{\medskip}
\frac{1}{2}\,a_{4} \left( b_{5}b_{3}+b_{5}b_{2}+b_{5}b_{1}-b_{4}b_{3}-b_{4}b_{2}-b_{4}b_{1}+2\,b_{5}b_{4}+2\,{a_{4}}^{2}-{b_{4}}^{2}+5\,{b_{5}}^{2} \right) \\\noalign{\medskip}
-2\,b_{2}{a_{1}}^{2}-b_{1}{a_{1}}^{2}-{a_{1}}^{2}b_{3}-{a_{1}}^{2}b_{4}-{a_{1}}^{2}b_{5}+{b_{1}}^{3}\\\noalign{\medskip}
4\,b_{1}{a_{1}}^{2}+5\,b_{2}{a_{1}}^{2}+{a_{1}}^{2}b_{3}+{a_{1}}^{2}b_{4}+{a_{1}}^{2}b_{5}+{a_{2}}^{2}b_{1}+b_{2}{a_{2}}^{2}-{a_{2}}^{2}b_{4}-{a_{2}}^{2}b_{5}+{b_{2}}^{3}\\\noalign{\medskip}
2\,b_{2}{a_{2}}^{2}-{a_{2}}^{2}b_{1}+3\,b_{3}{a_{2}}^{2}+{a_{2}}^{2}b_{4}+{a_{2}}^{2}b_{5}-2\,b_{4}{a_{3}}^{2}-{a_{3}}^{2}b_{1}-{a_{3}}^{2}b_{2}-3\,b_{3}{a_{3}}^{2}+{a_{3}}^{2}b_{5}+{b_{3}}^{3}\\\noalign{\medskip}
{a_{3}}^{2}b_{1}+{a_{3}}^{2}b_{2}-b_{4}{a_{3}}^{2}-{a_{3}}^{2}b_{5}+4\,b_{5}{a_{4}}^{2}+{a_{4}}^{2}b_{1}+{a_{4}}^{2}b_{2}+{a_{4}}^{2}b_{3}+5\,b_{4}{a_{4}}^{2}+{b_{4}}^{3}\\\noalign{\medskip}
-2\,b_{4}{a_{4}}^{2}-{a_{4}}^{2}b_{1}-{a_{4}}^{2}b_{2}-{a_{4}}^{2}b_{3}-b_{5}{a_{4}}^{2}+{b_{5}}^{3}
\end{array}
\right)
\]

%------------------------------------------------------------------------------%
\section{Generalizations}
%------------------------------------------------------------------------------%
We conclude with two possible generalizations of the systems considered.

There is also a periodic version of the system which is also integrable.
To obtain it we make some simple modifications in the various definitions.  In this section we indicate briefly the modifications needed to obtain the periodic version of the system.

The last root, $\Ga_N$, is given by:
	\be
	\Ga_N(h_{\Ga_j}) =
	\left\{
	\begin{array}{rll}
		     \ii& & j = 1,\\
		(-1)^N\; \ii& & j = N,\\
		0&      & \textrm{otherwise}.
	\end{array}
	\right.
	\ee
Alternatively the simple roots correspond to vectors of size $N$ as follows:
\be
\begin{array}{rcl}
	\Ga_{1} &=& (-\ii,-\ii,0, \dots, 0),\\
	\Ga_{2} &=& (0,\ii,\ii,0, \dots, 0),\\
	\Ga_{3} &=& (0,0,-\ii,-\ii,0, \dots, 0),\\
	    &\vdots&\\
	\Ga_N &=& (\ii,0, \dots, 0,(-1)^N\ii).
\end{array}
\ee

 Again  we obtain a Lax equation $\dot L(t) = [B(t), L(t)]$  which is consistent and gives the following equations for $1 \leqslant i \leqslant N$:
\be \label{dot2}
\begin{array}{lcl}
 \dot a_i	&=& a_i \, (b_{i+1} -b_i ), \quad \textrm{ for } i \neq N,\\
 \dot a_N	&=& a_N \, (b_1 -b_N ),\\[2ex]
 \dot b_i	&=& 2 \, ( a_i^2 - a_{i-1}^2 ), \quad \textrm{ for } i \nin \{1,m,m+1,N\},\\
 \dot b_1	&=& 2 \, ( a_1^2 + a_{N}^2 )\\
 \dot b_m	&=& -2 \, ( a_m^2 + a_{m-1}^2 )\\
 \dot b_{m+1} &=& 2 \, ( a_{m+1}^2 + a_{m}^2 )\\
 \dot b_N	&=& -2 \, ( a_N^2 + a_{N-1}^2 ) \ .
\end{array}
\ee

In the periodic case in  ${\bf (q,p)}$ coordinates we use the Hamiltonian

\be
%	H(q_1, \dots, q_N, \, p_1, \dots, p_N) =
	H(\mathbf{q},\mathbf{p}) =
		\sum_{i=1}^N \, \frac{1}{2} \, p_i^2 +
		\left(\sum _{i=1, i \neq m}^{N-1} \, e^{ q_i-q_{i+1}}\right) - e^{ q_m-q_{m+1}} -e^{ q_N-q_{1}} \ ,
		\label{a2}
\ee
which is slightly different from the classical, periodic Toda lattice.
(The minus sign at position $m$ and the last term.)
In order to obtain equations (\ref{dot2}), in the variables ${\bf (a,b)}$, we extend the Flaschka type transformation given in (\ref{FlaschkaTransformation}) by adding a new variable $a_N$.
\be
	a_N = \frac{1}{2}e^{ \frac{1}{2} (q_N-q_1)}. \label{FlaschkaTransformationPeriodic}
\ee

The Lax pair $\dot L = [B, L]$  can be written using the Kronecker product as follows:
\be \label{new L periodic}
 L= L_1 \otimes
  \begin{pmatrix}
  0 & -1 \cr
  1 & 0
  \end{pmatrix}
  +
  L_2 \otimes
  \begin{pmatrix}
  1 & 0 \cr
  0 & -1
  \end{pmatrix}
  +
  L_3 \otimes
  \begin{pmatrix}
  1 &  0 \cr
  0  & (-1)^{N+1}
  \end{pmatrix}
\ee
where $L_1$ and $L_2$ are as in (\ref{new L1}) and (\ref{new L2}) respectively and

\be \label{new L3 periodic}
L_3 =
\begin{pmatrix}
  0       & \cdots &    & \cdots & 0    & (-1)^{N}a_N \cr
  \vdots    & \ddots &    &    &     & 0       \cr
         &    &    &    &     & \vdots    \cr
  \vdots    &    &    &    &     &        \cr
  0       &    &    &    & \ddots  & \vdots    \cr
  (-1)^{N}a_N & 0   & \cdots &    & \cdots  & 0
\end{pmatrix}
\ .
\ee

The matrix $B$ is defined as follows:
\be \label{new B periodic}
 B= B_1 \otimes
  \begin{pmatrix}
  0 & 1 \cr
  1 & 0
  \end{pmatrix}
  +
  B_2 \otimes
  \begin{pmatrix}
  0 & 1 \cr
  (-1)^N & 0
  \end{pmatrix}
\ee
where $B_1$ is as in (\ref{new B1}) and

\be \label{new B2 periodic}
B_2 =
\begin{pmatrix}
  0     & \cdots &    & \cdots & 0    & a_N   \cr
  \vdots  & \ddots &    &    &     & 0    \cr
       &    &    &    &     & \vdots \cr
  \vdots  &    &    &    &     &     \cr
  0     &    &    &    & \ddots  & \vdots \cr
  (-1)^Na_N & 0   & \cdots &    & \cdots  & 0
\end{pmatrix} \ .
\ee

There is also a different Lax pair $(M, A)$

\begin{equation}
M =
\begin{pmatrix}
	b_1 &  a_1 & 0 & \cdots & \cdots &\ii a_N \cr
	a_1 & b_2 & a_2 & \cdots &    & \vdots \cr
	0 & a_2 & b_3 & \ddots &  &  \cr
	\vdots & & \ddots &  \ddots &\ii a_m & \vdots \cr
	\vdots & & & \ii a_m & \ddots & a_{N-1} \cr
	\ii a_N & \cdots & & \cdots & a_{N-1} & b_N
\end{pmatrix}
                        \ ,
\end{equation}
i.e. we replace $a_m$ with $\ii a_m$.  We define

 \begin{equation}
A =
\begin{pmatrix}
	0 &  a_1 & 0 & \cdots & \cdots & -\ii a_N \cr
	-a_1 & 0 & a_2 & \cdots &    & \vdots \cr
	0 & -a_2 &0 & \ddots &  &  \cr
	\vdots & & \ddots &  \ddots &\ii a_m & \vdots \cr
	\vdots & & & -\ii a_m & \ddots & a_{N-1} \cr
	\ii a_N & \cdots & & \cdots &- a_{N-1} &0
\end{pmatrix}
                        \ .
\end{equation}

One verifies easily that the Lax pair is equivalent to the equations (\ref{dot2}).

Note that the Hamiltonian has one more term in $(a,b)$ coordinates i.e.
\be \label{H2 periodic}
H_2	= H(q_1, \dots, q_N, \, p_1, \dots, p_N)
	= \frac{1}{2}\sum_{i=1}^N b_i^2+ \left(\sum_{i=1, i \neq m}^{N-1} a_i^2 \right)-a_m^2 -a_N^2
\ee

The Poisson bracket $\Gp_1$ for the periodic case is defined as in (\ref{poisson1}) with the addition of
\be
	\{a_N, b_1 \} = a_N. \label{poisson p1 periodic}
\ee

There are two Casimirs
	$H_1=b_1+b_2 + \dots +b_N$ and
	$a_1 a_2 \cdots a_N$.

There is a quadratic Toda bracket $\Gp_2$ with defining relations as in (\ref{poisson p2}) with the addition of
\be \label{poisson p2 periodic}
\begin{array}{rcll}
	\{a_N, a_1 \}		&=& -\frac{1}{2} a_N a_1 \ ,\\
	\{a_N, b_N \}		&=& -a_N b_N \ ,\\
	\{a_N, b_1 \}		&=& a_N b_1 \ ,\\
	\{b_N, b_1 \}		&=& -2\, a_N^2 \ ;
\end{array}
\ee

This bracket also has two Casimirs:
	${\rm det} \, L$ and $a_1 a_2 \cdots a_N$.
	
$H_1 = b_1 + \dots b_N$ is the Hamiltonian.

Furthermore, $\pi_2$ is compatible with $\pi_1$.
We also have the Lenard relation
\be
\pi_1 \nabla H_2 = \pi_2 \nabla H_1 \ . \label{a98 periodic}
\ee

We can generalize these systems in a different way by choosing different signs in the potential terms.
Tomei in \cite{tomei} studies the topology of the set of real, symmetric, tridiagonal $n\times n$ matrices with fixed eigenvalues
\bd \lambda_1 > \lambda_2 > \dots > \lambda_N \ .  \ed
This set consists of $2^{N-1}$  components depending on the signs ($\pm$) of the variables $a_1, a_2, \dots, a_N$. We may construct a Toda system  corresponding to each component and any two of them are isomorphic via a complex transformation. In this paper we have  studied in detail one such example but in general the number of such systems is $2^{N-1}$. It is also possible to produce two Lax pairs for each system. The first one is obtained from the Flaschka Lax pair by replacing some of the $a_j$ by $\ii a_j$. The second Lax pair which is obtained in $so(p,q)$ is also easily obtained by modifying some of the signs of (\ref{new L}).
We illustrate with an example for $N=3$.
\begin{example}
In this example we display all possible Toda systems in the case of $N=3$.  We give a Lax pair and the equations of motion for each system. For the last system we indicate also the second Lax pair which is obtained by an obvious change of signs in (\ref{new L}).
\begin{equation*}
L_1  =
	\begin{pmatrix}
	b _{1} & a _{1} & 0 \\
	a _{1} & b _{2} & a _{2}  \\
	0 & a _{2} &  b _{3} \\
	\end{pmatrix}
\quad \quad
B_1  =
	\begin{pmatrix}
	0 & a _{1} & 0 \\
	-a _{1} & 0 & a _{2}  \\
	0 & -a _{2} &  0 \\
	\end{pmatrix}
\end{equation*}
Equations of motion:
\[\displaystyle
\begin {array}{rclccrcl}
\dot{b_1}(t)&=&2\,(a_{1}^{2})&&&
\dot{a_1}(t)&=& (b_{2}-b_{1}) a_{1}\\\noalign{\medskip}
\dot{b_2}(t)&=&2\,(a_{2}^{2}-a_{1}^{2})&&&
\dot{a_2}(t)&=& (b_{3}-b_{2}) a_{2}\\\noalign{\medskip}
\dot{b_3}(t)&=&2\,(-a_{2}^{2})
\end {array}\]
\end{example}

\begin{example}
\begin{equation*}
L_2  =
	\begin{pmatrix}
	b _{1} & a _{1} & 0 \\
	a _{1} & b _{2} & \ii a _{2}  \\
	0 & \ii a _{2} &  b _{3} \\
	\end{pmatrix}
\quad \quad
B_2  =
	\begin{pmatrix}
	0 & a _{1} & 0 \\
	-a _{1} & 0 & \ii a _{2}  \\
	0 & -\ii a _{2} &  0 \\
	\end{pmatrix}
\end{equation*}
Equations of motion:
\[\displaystyle
\begin {array}{rclccrcl}
\dot{b_1}(t)&=&2\,(a_{1}^{2})&&&
\dot{a_1}(t)&=& (b_{2}-b_{1}) a_{1}\\\noalign{\medskip}
\dot{b_2}(t)&=&2\,(-a_{2}^{2}-a_{1}^{2})&&&
\dot{a_2}(t)&=& (b_{3}-b_{2}) a_{2}\\\noalign{\medskip}
\dot{b_3}(t)&=&2\,(a_{2}^{2})
\end {array}\]
\end{example}

\begin{example}
\begin{equation*}
L_3  =
	\begin{pmatrix}
	b _{1} & \ii a _{1} & 0 \\
	\ii a _{1} & b _{2} & a _{2}  \\
	0 & a _{2} &  b _{3} \\
	\end{pmatrix}
\quad \quad
B_3  =
	\begin{pmatrix}
	0 & \ii a _{1} & 0 \\
	-\ii a _{1} & 0 & a _{2}  \\
	0 & -a _{2} &  0 \\
	\end{pmatrix}
\end{equation*}
Equations of motion:
\[\displaystyle
\begin {array}{rclccrcl}
\dot{b_1}(t)&=&2\,(-a_{1}^{2})&&&
\dot{a_1}(t)&=& (b_{2}-b_{1}) a_{1}\\\noalign{\medskip}
\dot{b_2}(t)&=&2\,(a_{2}^{2}+a_{1}^{2})&&&
\dot{a_2}(t)&=& (b_{3}-b_{2}) a_{2}\\\noalign{\medskip}
\dot{b_3}(t)&=&2\,(-a_{2}^{2})
\end {array}\]
\end{example}

\begin{example}
\begin{equation*}
L_4  =
	\begin{pmatrix}
	b _{1} & \ii a _{1} & 0 \\
	\ii a _{1} & b _{2} & \ii a _{2}  \\
	0 & \ii a _{2} &  b _{3} \\
	\end{pmatrix}
\quad \quad
B_4  =
	\begin{pmatrix}
	0 & \ii a _{1} & 0 \\
	-\ii a _{1} & 0 & \ii a _{2}  \\
	0 & -\ii a _{2} &  0 \\
	\end{pmatrix}
\end{equation*}
Equations of motion:
\[\displaystyle
\begin {array}{rclccrcl}
\dot{b_1}(t)&=&2\,(-a_{1}^{2})&&&
\dot{a_1}(t)&=& (b_{2}-b_{1}) a_{1}\\\noalign{\medskip}
\dot{b_2}(t)&=&2\,(-a_{2}^{2}+a_{1}^{2})&&&
\dot{a_2}(t)&=& (b_{3}-b_{2}) a_{2}\\\noalign{\medskip}
\dot{b_3}(t)&=&2\,(a_{2}^{2})
\end {array}\]
\end{example}
Alternative Lax pair for the last example:
\[\displaystyle
L_4'\, : \, \left(
\begin{array}{cccccc}
0&-b_1&-a_1&0&0&0\\\noalign{\medskip}
b_1&0&0&a_1&0&0\\\noalign{\medskip}
-a_1&0&0&b_2&a_2&0\\\noalign{\medskip}
0&a_1&-b_2&0&0&-a_2\\\noalign{\medskip}
0&0&a_2&0&0&-b_3\\\noalign{\medskip}
0&0&0&-a_2&b_3&0\\\noalign{\medskip}
\end{array}
\right)
\quad
B_4'\, : \, \left(
\begin {array}{cccccccccc} 0&0&0&a_1&0&0\\\noalign{\medskip}
0&0&a_1&0&0&0\\\noalign{\medskip}
0&a_1&0&0&0&a_2\\\noalign{\medskip}
a_1&0&0&0&a_2&0\\\noalign{\medskip}
0&0&0&a_2&0&0\\\noalign{\medskip}
0&0&a_2&0&0&0\\\noalign{\medskip}
\end {array} \right)
\]

\end{document}